\documentclass{lmcs}
\pdfoutput=1

\usepackage{lastpage}
\renewcommand{\dOi}{14(3:6)2018}
\lmcsheading{}{1--\pageref{LastPage}}{}{}%
{May~31,~2018}{Jul.~31,~2018}{}

\usepackage[utf8]{inputenc}
\usepackage[T1]{fontenc}
\usepackage[english]{babel}
\usepackage{etex}
\usepackage{algorithm}
\usepackage{algorithmic}
\usepackage{xcolor}
\usepackage{multicol} 
\usepackage{mathtools}
\usepackage{amssymb}
\usepackage{microtype}
\usepackage{csquotes}
\usepackage{xspace}
\usepackage{hyperref}
\hypersetup{hidelinks}
\usepackage{relsize}
\usepackage{fnpct}
\usepackage{tikz}
\usetikzlibrary{calc,arrows,matrix,positioning,fit,decorations.markings}
\newcommand{\ptext}{$\mathrel{\mathsmaller\bullet}$\xspace} %
\newcommand{\symbf}{\mathtt}				%
\newcommand{\extinflist}[1]{#1,\,\dots}						%
\newcommand{\lcalc}{$\lambda$-calculus\xspace}
\newcommand{\langprog}[1]{\textsc{#1}}
\newcommand{\tpower}[1]{^{\otimes #1}}
\newcommand{\ptensor}{\!\mathrel{\dot{\otimes}}\!}
\newcommand{\unitalg}{\C I}
\newcommand{\unit}{I}

\newcommand{\comp}{\TT{Comp}}

\newcommand{\rgt}{\TT{r}}
\newcommand{\lft}{\TT{l}}
\newcommand{\set}[2]{\left\{#1 \ \middle|\ #2\right\}}

\newcommand{\lang}{\mathcal{L}}

\newcommand{\closedh}{\BB T}
\newcommand{\ualg}{\C U}

\newcommand{\start}{\star}
\newcommand{\IO}{\TT{LR}}
\newcommand{\p}{\!\mathrel{\textstyle\mathsmaller\bullet}\!}

\newcommand{\concrete}{+}

\newcommand{\move}{\texttt{MOVE}\xspace}
\newcommand{\swap}{\texttt{SWAP}\xspace}

\newcommand{\C}[1]{\mathcal{#1}}
\newcommand{\F}[1]{\mathfrak{#1}}
\newcommand{\BB}[1]{\mathbb{#1}}

\newcommand{\TT}[1]{\mathtt{#1}}

\newcommand{\void}{\varnothing}

\newcommand{\flow}{\leftharpoonup}

\newcommand{\sflow}{\leftrightharpoons}

\newcommand{\trans}[2]{\xrightarrow[#2]{#1}}

\newcommand{\lobs}{$\F S$-observation\xspace}

\definecolor{bluegray}{rgb}{0.4, 0.6, 0.8}

\newcommand{\compclass}[1]{\textsc{#1}} %

\newcommand{\Logspace}{\compclass{Log\-space}\xspace}						%
\newcommand{\NLogspace}{\compclass{NLog\-space}\xspace}				%
\newcommand{\NorDLogspace}{\compclass{(N)Log\-space}\xspace}
\newcommand{\coNLogspace}{\compclass{co-NLog\-space}\xspace}	%
\newcommand{\Ptime}{\compclass{Ptime}\xspace}								%

\DeclareMathOperator{\Id}{Id}

\DeclareMathOperator{\var}{Var}

\newcommand{\Lcalc}{{\hbox{$\lambda$-cal}culus}\xspace}

\newcommand{\Lterms}{$\lambda$-terms\xspace}
\newcommand{\GofI}{geometry of interaction\xspace}

\newcommand{\locit}[1]{#1} %
\newcommand{\eg}{\locit{e.g.}~}
\newcommand{\ie}{\locit{i.e.}~}
\newcommand{\via}{\locit{via}~}

\newcommand{\etc}{\locit{etc.}\xspace}

\newcommand{\incise}[1]{---#1---} %
\begin{document}

\title{Unification and Logarithmic Space}
\author[C. Aubert]{Clément Aubert}
\address{School of Cyber and Computer Sciences, Augusta University, Augusta, GA, USA}
\email{clement.aubert@math.cnrs.fr}
\urladdr{\url{https://aubert.perso.math.cnrs.fr/}}
\thanks{This work was partly supported by the ANR-11-INSE-0007 REVER, the ANR-11-BS02-0010 Récré and the ANR-10-BLAN-0213 Logoi.}

\author[M. Bagnol]{Marc Bagnol}
\address{Laboratoire d'Informatique du Parallélisme, ENS Lyon, France}
\email{marc.bagnol@ens-lyon.fr}
\urladdr{\url{http://www.normalesup.org/~bagnol/}}
%


\begin{abstract}
	We present an algebraic characterization of the complexity classes \Logspace and \NLogspace, using an algebra with a composition law based on unification.
	This new bridge between unification and complexity classes is rooted in proof theory and more specifically linear logic and \GofI.
	
	We show how to build a model of computation in the unification algebra
	and then, by means of a syntactic representation
	of finite permutations in the algebra,
	we prove that whether an observation (the algebraic counterpart of a program) 
	accepts a word can be decided within logarithmic space.
	Finally, we show that the construction naturally corresponds to pointer machines, a convenient
	way of understanding logarithmic space computation.
\end{abstract}

\subjclass{F.1.3 Complexity Measures and Classes, F.4.1 Mathematical
	Logic.
	\\
	\textit{\href{http://www.acm.org/about/class/class/2012}{2012 ACM Subject Classification}}:
	[\textbf{Theory of computation}]
	Models of computation — Abstract machines;
	Computational complexity and cryptography — Complexity theory and logic;
	Logic — Proof theory;
	Semantics and reasoning — Program semantics — Algebraic semantics.
}
\titlecomment{This is an extended version of a preceding work~\cite{Aubert2014}.}
\maketitle
 
\section*{Introduction}
\subsection*{Proof Theory and Implicit Complexity Theory}
Complexity theory classifies the difficulty of problems by studying the asymptotic bounds on the 
resources (time, memory, processors, etc.) needed by a model of computation to run a program that solves them.
It was originally dependent on models of computation 
(Turing machines, random access machines, Boolean circuits, etc.), associated to a reasonable cost-model, that ran the implementation of an algorithm, a program.
The aim of implicit computational complexity (ICC) theory is to abstract away the specificities of 
hardware by focusing on the way programs are written.
For instance, weaker recursion schemata~\cite{Bellantoni1992a}, stratified 
recurrence~\cite{Leivant1993}, or quasi-interpretation~\cite{Bonfante2011} restrict expressivity
\via internal limitations on programming languages or function algebras rather than on available resources.

There is a longstanding tradition of relating proof theory (more specifically linear logic~\cite{Girard1987}) and implicit complexity theory, thanks to the Curry-Howard\incise{or \emph{proofs as programs}}correspondence.
Indeed, mathematical proofs and typed programs, both endowed with an evaluation mechanism (respectively cut-elimination and execution), are viewed as isomorphic, so that restrictions on the former translate seamlessly into limitations on the latter.
Fragments of linear logic\incise{bounded~\cite{Girard1992,Lago2010}, elementary~\cite{Danos2003}, light~\cite{Girard1995} or stratified~\cite{Schopp2007,Baillot2010}, to name a few}were proven to characterize complexity classes.
By removing or restricting rules of derivation, one excludes proofs and henceforth algorithms: the class of programs accepted can then be proven to (extensionally) correspond to functions of a certain complexity class.
In these restricted logics, the cut-elimination procedure\incise{which represents execution of programs
as rewriting of proofs}is simpler, and problems that are undecidable in general (such as termination of computation)
can become of a manageable complexity.

\subsection*{Geometry of Interaction}
The study of cut-elimination has grown to a central topic in proof theory and as a consequence
its mathematical modelling became of great interest.
The \GofI~\cite{Girard1989b} research program led to mathematical models of cut-elimination in terms of paths~\cite{Asperti1994}, token machines~\cite{Laurent2001}, operators algebras~\cite{Girard1989a,Girard2011a,Seiller2018}
or graphs~\cite{Danos1990,Seiller2012b,Seiller2012a}.
The general perspective is to consider untyped objects modelling untyped programs
and to represent algebraically cut-elimination.

This approach was already used with complexity concerns~\cite{Baillot2001,Girard2012,Aubert2016IaC, Aubert2016mscs}\footnote{For a more advanced discussion on the \enquote{sister approaches} relying on the theory of von Neumann algebras~\cite{Girard2012,Aubert2016mscs,Aubert2016IaC} one should refer to the \enquote{related work} section, page \pageref{related-works}.}.
It differs from usual ICC \via proof theory because the restrictions on the expressivity of programs is not obtained through restrictions on type systems.
Instead, limitations imposed on the objects representing proofs rule out computational principles 
on the semantics side and allow to capture complexity classes.
This enables the use of methods coming from all areas of mathematics: for instance, an action of 
the group of permutations on an unbounded tensor product will provide us with our basic computational principle.

\subsection*{Unification} 
Unification is one of the key-concepts of theoretical computer science, for it is used in logic programming and is a classical subject of study for complexity theory.
Its different names and variants\incise{unification, matching, resolution rule, etc.}always comes down to the same question: is there a substitution to make two first-order terms equal?
It is an interesting mechanism of computation that can be seen as more primitive than other evaluation procedures
such as the $\beta$-reduction of \Lcalc.

The resolution rule of logic programming~\cite{Robinson1965} serves as a basis for a more syntactical version of \GofI~\cite{Girard1995a,Girard2013}, where cut-elimination is represented as 
iterated matching in a \emph{unification algebra}~\cite{Bagnol2014}.
In this setting, proofs are represented as clauses (or \enquote{flows}), which have a natural notion of size, height, \etc and relate closely to the study of complexity of
logic programming~\cite{Dantsin2001}.
This is an intuitive framework, yet expressive enough for our purposes.

\subsection*{Contribution and Outline of the Article}
We carry on the methodology of bridging \GofI and complexity theory with this renewed approach.
It relies on a simple representation of execution in a unification-based algebra, defined in Section~\ref{sec_unification}, that is shown to represent some algebraic structures syntactically.

In Section~\ref{sec_words}, we present the framework where computation takes place and show how to represent data and programs.
Inputs are considered to be words over a finite alphabet, encoded thanks to the classical Church
representation of lists (Section~\ref{sec_words_and_obs}).
This raises a question about invariance up to different representations of the same input, addressed in Section~\ref{sec_normativity}.

This construction is finally specialized in Section~\ref{sec_logspace} to a subalgebra
relying on a representation of permutations in the unification algebra.
The soundness of the construction with respect to logarithmic space computation (both
deterministic and non-deterministic) is proven thanks to a procedure deciding the outcome of the interaction of the representation of a program with the representation of a data.
Observations are the algebraic counterpart of programs, and they are shown in Section~\ref{subsec_completness} to correspond to a natural notion of read-only Turing machines: pointer machines. In that perspective the
algebraic notion of isometricity will correspond to reversibility of computation.%

\section{The Unification Algebra}
\label{sec_unification}
\subsection{Unification}
Unification can be thought of as the study of formal solving of equations between terms.
This topic was introduced by Herbrand~\cite{Herbrand1930}, but became really widespread after the work of J.~A.~Robinson~\cite{Robinson1965}
on automated theorem proving.
The unification technique is also at the core of the logic programming language \langprog{Prolog}
and type inference for functional programming languages such as \langprog{CaML} and \langprog{Haskell}.

\begin{nota}
	We consider first-order terms, written $\extinflist{t, u, v}$, built from variables, noted in italics font (\eg $x,y$),
	and function symbols with an assigned finite arity, written in typewriter font
	(\eg $\symbf c$, $\symbf f(\cdot)$, $\symbf g(\cdot,\cdot)$).
	Symbols of arity $0$ will be called \emph{constants}.
	
	Sets of variables and of function symbols of any arity are supposed infinite.
	We distinguish a binary function symbol \ptext (in infix notation) and a constant symbol $\start$.
We will omit the parentheses for \ptext and write $t \p u\p v$ for $t\p(u\p v)$.	
	
	We write $\var(t)$ the set of variables occurring in the term $t$ and say that $t$ is \emph{closed} if $\var(t)=\void$. %
We will write $\theta t$ the result of applying in parallel the substitution $\theta$, written $\{t_1 \mapsto u_1 ; t_2 \mapsto u_2 ; \hdots\}$, to the term $t$.
\end{nota}

\begin{defi}[renaming and instance]
\label{renaming}
A \emph{renaming} is a substitution $\alpha$ that bijectively maps variables to variables.
A term $t'$ is a \emph{renaming} of $t$ if $t'=\alpha t$ for some renaming $\alpha$.
Two substitutions $\theta$, $\psi$ are equal \emph{up to renaming} if there is a renaming
$\alpha$ such that $\psi=\alpha\theta$.

A substitution $\psi$ is an \emph{instance} of $\theta$ if there is a substitution
$\sigma$ such that $\psi=\sigma\theta$.
\end{defi}

\begin{exas}
Let
\[\alpha = \{x \mapsto y ; y \mapsto x\} \quad \theta = \{x \mapsto \symbf c ; z \mapsto \symbf g(\start)\} \quad \psi = \{y \mapsto \symbf c ; z \mapsto \symbf g(\start)\}\]
be three substitutions and 
\[t = \symbf f(x) \p z \quad t' = \symbf f(y) \p z\]
be two terms.
Then $\alpha$ is a renaming, $t'$ is a renaming of $t$, and $\theta$ and $\psi$ are equal up to renaming.
As $\var(t) = \{x, z\}$, $\theta t = \symbf f(\symbf c) \p \symbf g(\start)$ is a closed term.
\end{exas}

\begin{defi}[unification]
\label{def_unification}
Two terms $t,u$ are \emph{unifiable} if there is a substitution $\theta$ such
that $\theta t=\theta u$.
We say that $\theta$ is a \emph{most general unifier (MGU) of $t$, $u$} if any other unifier of $t$, $u$ is an instance of $\theta$.
\end{defi}

\begin{rem}
\label{renaming-MGU}
It is easy to check that any two MGU of a pair of terms are equal up to renaming.
\end{rem}

We will be interested mostly in the weaker variant of unification where one can first perform
renamings on terms to make their variables distinct.
We therefore introduce a specific vocabulary for it.

\begin{defi}[disjointness and matching]
\label{disjoint}
Two terms $t$, $u$ are \emph{matchable} if $t'$, $u'$ are unifiable, where $t'$, $u'$ are renamings (Definition~\ref{renaming}) of $t$, $u$ such that $\var(t')\cap\var(u')=\varnothing$.

If two terms are not matchable, they are said to be \emph{disjoint}.
\end{defi}

\begin{exa}
The terms $x$ and $\symbf c \p x$ are not unifiable,
but they are matchable, as a renaming of $x$, for instance $\alpha x = \{x\mapsto y;y\mapsto x\}x=y$, is unifiable with $\symbf c \p x$.
\end{exa}

A fundamental result on first-order unification is the (decidable) existence of most general unifiers
in cases where the unification problem has a solution.

\begin{prop}[MGU]
If $t$ and $u$ are unifiable, they have a MGU.
Whether two terms are unifiable and, in case they are, finding a MGU is a decidable problem.
\end{prop}

As unification grew in importance, the study of its complexity gained in attention.
A complete survey~\cite{Knight1989} tells the story of the bounds getting sharpened: general first-order
unification was finally proved~\cite{Dwork1984} to be a \Ptime-complete problem.

In this article, we will consider a much simpler case of the problem: matching a term against a closed term, which has been shown to be tractable within deterministic logarithmic space.

\begin{thm}[Matching is in \Logspace \protect{\cite[p.~49]{Dwork1984}}]
\label{thm-unif-logspace}
Deciding if two terms $t$, $u$ with $u$ closed are unifiable and, if so, producing their MGU, is in \Logspace.
\end{thm}

This result is stronger than what we will eventually need: in the proof of Theorem~\ref{soundness}, 
where we will unify what will be called a flow against a closed terms, we will need to run the matching operation 
within linear space to keep the whole algorithm in logarithmic space.

\subsection{Flows and Wirings}
We now design an algebra with a product based on unification.
Let us start by setting up a monoid with a partially defined product, which will be the basis of the construction.

\begin{defi}[flows]
\label{def_flow}
A \emph{flow} is an oriented pair of first-order terms $t\flow u$ such that $\var(t)=\var(u)$.

Flows are considered up to renaming: for any renaming $\alpha$,
$t\flow u=\alpha t\flow \alpha u$.

We set $\unit:=x\flow x$ and $(t\flow u)^{\dagger}:=u\flow t$, so that $(.)^{\dagger}$ is an involution.%
\end{defi}

A flow $t\flow u$ can be thought of as a \texttt{match…with u -> t} in a ML-style language
or as a specific kind of Horn clause\footnote{The precise connections with logic programming, that comes with a relaxed definition of flows, were subsequently exposed~\cite{Aubert2014b,%
Aubert2016fossacs%
}.}.
The composition of flows follows this intuition: it is an instance of the resolution rule of
logic programming.

\begin{defi}[product of flows]
Let $u\flow v$ and $t\flow w$ be two flows.
Suppose we have chosen two representatives of the renaming classes such that their sets of variables are disjoint.
The \emph{product} of $u\flow v$ and $t\flow w$ is defined if $v,t$ are unifiable with MGU $\theta$
(the choice of a MGU does not matter because of Remark~\ref{renaming-MGU}) and in that case:
\[(u\flow v)(t\flow w):=\theta u \flow \theta w\]
\end{defi}

\begin{defi}[action on closed terms]
\label{flow-action}If $t$ is a closed term, $(u\flow v)(t)$ is defined whenever $v$ and $t$ are unifiable, with MGU $\theta$, in that case $(u\flow v)(t):=\theta u$
\end{defi}

\begin{exas}
Composition of flows: $(x\p \symbf c \flow (\symbf c\p\symbf c) \p x)(y\p z\flow z\p y)= x\p\symbf c \flow x\p\symbf c\p\symbf c$.

Action on a closed term: $(x\p\symbf c \flow x\p\symbf c\p\symbf c)(\symbf d\p\symbf c\p\symbf c)=\symbf d\p\symbf c$.
\end{exas}

\begin{rem}
The condition on variables ensures that the result of an action on a closed term is a closed term, because $\var(u)\subseteq\var(v)$, and that the action is injective on its definition domain, because $\var(v)\subseteq\var(u)$.

Moreover, the action is compatible with the product of flows: for $l$ and $k$ two flows and $t$ a term, $l(k(t))=(lk)(t)$ and both are defined at the same time.
\end{rem}

By adding a formal element $\bot$ (representing the failure of unification) to the set of flows, one could turn the product into a completely defined operation, making the set of flows an \emph{inverse monoid}.
However, we will need to consider the wider algebra of \emph{sums} of flows that is easily defined directly from the partially defined product.
An analogy can give an insight on this need:
when considering logic programs, one wants to manipulate \emph{set of clauses} and not only a single clause.
All the same, we want here to compute thanks to sets of flows, formally represented as sums of flows in our algebra, although the coefficient will play a minor role (we will essentially take them all to be $1$, as detailed in Definition~\ref{concrete}).

We therefore now lift the structure to a $\ast$-algebra by considering formal sums of flows with complex
coefficients and extending all our operations by linearity.

\begin{defi}[wirings, unification algebra]
\label{wir-and-unif-alg}
\emph{Wirings} are $\BB C$-linear combinations of flows
endowed with the following operations (with
$\lambda_i, \mu_j \in \BB C$, $l_i, k_j$ two flows
and $\overline\lambda$ the complex conjugate of $\lambda$):
\begin{align*}
\bigg(\sum_i \lambda_il_i\bigg) \bigg(\sum_j \mu_jk_j\bigg) & :=
\sum_{\mathclap{\substack{i,j \text{ such that} \\ (l_ik_j)\text{is defined}}}}\lambda_i\mu_j(l_ik_j)\\
\bigg(\sum_i \lambda_il_i\bigg)^{\dagger} & :=\sum_i \overline\lambda_il_i^{\dagger}
\end{align*}

We write $\ualg$ the set of wirings and refer to it as the \emph{unification algebra}.
\end{defi}

\begin{rem}
Indeed, $\ualg$ is a unital $\ast$-algebra: it is a $\BB C$-algebra, considering the product defined above, with an involution $(.)^{\dagger}$ and a unit $\unit$ (Definition~\ref{def_flow}).
\end{rem}

Let us note here that the full $\BB C$-linear framework is not strictly necessary to obtain our complexity results, but is kept to stay in line with previous work on this topic and does not overly complexify the presentation.
An account of this work where complex coefficients have been dropped and the linear sum is replaced by union of sets can be found in the second author's PhD thesis~\cite{Bagnol2014}.

To study computation and concrete programs (which do not involve complex coefficients)
we need to restrict the large algebraic framework defined and to consider wirings whose only
coefficient is $1$.

\begin{defi}[concrete wirings]
\label{concrete}
A wiring is \emph{concrete} whenever it is a sum of flows with all coefficients equal to~$1$.
Given a set of wirings $E$ we write $E^+$ the set of all concrete wirings of $E$, and will omit to write the coefficients.
\end{defi}

We can then consider further the notion of \emph{isometric wiring}.
As they act on closed terms as a partial injection (Lemma~\ref{lem_isom}), they can be considered as behaving in a reversible way.
On the other hand, they satisfy the algebraic property of partial isometries, that is
$WW^\dagger W=W$.

\begin{defi}[isometric wiring]
\label{piso}
A concrete wiring $\sum_i u_i \flow t_i$ is \emph{isometric}
if the $u_i$ are pairwise disjoint (Definition~\ref{disjoint}) and the $t_i$ are pairwise disjoint.
\end{defi}

\begin{exa}
The sum of flows $(\symbf c\p x \flow x\p\symbf d)+(\symbf d\p\symbf c\flow \symbf c\p \symbf c)$ is
an isometric wiring.
Note that a wiring containing a single flow will always be a partial isometry.
\end{exa}

It will be convenient to consider the action of wirings on closed terms, making them linear operators
on the vector space spawned by closed terms.
We therefore extend the definition of action on closed terms (Definition~\ref{flow-action}) to wirings.

\begin{defi}[$\closedh$, action on closed terms]
\label{act_vect}
Let $\closedh$ be the free $\BB C$-vector space spawned by closed terms.
Wirings act on base vectors of $\closedh$ in the following way:

\[\bigg(\sum_i \lambda_il_i\bigg)(t) :=\!\sum_{\mathclap{\substack{i \text{ such that } \\ l_i(t)\text{ is defined }}}}\lambda_i \big(l_i(t)\big) \ \ \in \closedh\]
which extends by linearity into an action on the whole $\closedh$.
\end{defi}

\begin{lem}[isometric action]
\label{lem_isom}
Let $F$ be an isometric wiring and $t$ a closed term.
We have that $F(t)$ and $F^{\dagger}(t)$ are either $0$ or another closed term $t'$ (seen as an element of $\closedh$).
It follows that any isometric wiring induces a partial injection on the set of terms.
\end{lem}

\begin{proof}
	A wiring $F$ is isometric if and only if $F^\dagger$ is so we can focus on $F(t)$: because
	$F=\sum_i u_i \flow t_i$ with the $t_i$ pairwise disjoint, then $t$ match at most one of the
	$t_i$ and therefore the sum $F(t)$ can contain at most one element, with coefficient $1$.
	
	Then the action of an isometric wiring on closed terms is a partial function, with a partial inverse
	given by the action of $F^\dagger$.
\end{proof}

\subsection{Tensor Product and Permutations}\label{permutation}
We now define the representation in the unification algebra $\ualg$ of structures that provide more expressivity.
Thanks to the notion of tensor product, we will build wirings and subalgebras that are
split into components computing independently. 
Unbounded tensor products will allow to represent a potentially unbounded number of data stores.
Finite permutations have a natural representation in the algebra that acts on the unbounded tensor
product, allowing to represent manipulation of these stores.

All this will provide enough room and computational principles to represent in Section~\ref{subsec_completness} a basic model of computation, with pointers and internal states.

\begin{defi}[tensor product]
\label{ptensor}
Let $u\flow v$ and $t\flow w$ be two flows.
Suppose we have chosen representatives of these renaming classes that have their sets of variables disjoint. We define their \emph{tensor product} as
$(u\flow v) \ptensor (t\flow w):= u\p t \flow v\p w$.
The operation is extended to wirings by bilinearity.

Given two $\ast$-algebras $\C A$, $\C B$, we define their tensor product as the $\ast$-algebra $\C A\ptensor \C B$
spawned by
\[\set{F\ptensor G}{F\in\C A,G\in \C B}\]
\end{defi}

This actually defines an embedding of the usual algebraic tensor product into $\ualg$, which means
in particular that $(F\ptensor G)(P\ptensor Q)=(FP)\ptensor(GQ)$.
As for \ptext, we will omit the parentheses for $\ptensor$ and write $\C A\ptensor\C B\ptensor\C C$ for $\C A\ptensor(\C B\ptensor\C C)$.

Once we have the basic tensor product as a building block, we can define the unbounded one by putting
together bigger and bigger tensor powers of the same $\ast$-algebra, with a variable in the end standing
for the fact that the size is not specified in advance.

\begin{defi}[unbounded tensor]
\label{unbounded}
Let $\C A$ be a $\ast$-algebra, we define the $\ast$-algebras $\C A\tpower n$ for all $n\in \BB N$ as
(letting $\unitalg:=\set{\lambda\unit}{\lambda\in \BB C}$,
with $\unit=x\flow x$ as in Definition~\ref{def_flow})
\[\C A\tpower 0:= \unitalg \quad\text{and}\quad \C A\tpower {n+1}:=\C A \ptensor \C A\tpower n\]
and the $\ast$-algebra $\C A\tpower \infty$ spawned by $\displaystyle\bigcup_{\mathclap{n\in \BB N}} \C A\tpower n\ $.
\end{defi}

We consider that finite permutations can be composed even when their domain of definition do not match,
and get a natural representation of them based on the binary function symbol \ptext.

\begin{defi}[representation]
To a permutation $\sigma \in \F S_n$ we associate the flow
\[[\sigma]:= x_1\p x_2\p\cdots\p x_n \p y\flow x_{\sigma(1)}\p x_{\sigma(2)}\p\cdots\p x_{\sigma(n)}\p y\]
\end{defi}

A permutation $\sigma\in\F S_n$ can act on the first $n$ components of the unbounded tensor product (Definition~\ref{unbounded}) by swapping them and leaving the rest unchanged.
The wirings $[\sigma]$ internalize this action: in the above definition, the variable $y$ at the end stands for the components that are not affected.

\begin{exa}
Let $\tau\in\F S_2$ be the permutation swapping the two elements of $\{1,2\}$ and $U_1\ptensor U_2\ptensor U_3\ptensor \unit \in \ualg\tpower 3\subseteq \ualg\tpower\infty$.

We have $[\tau]=x_1\p x_2\p y\flow x_2\p x_1\p y$ and $[\tau](U_1\ptensor U_2\ptensor U_3\ptensor \unit)[\tau]^{\dagger}=U_2\ptensor U_1\ptensor U_3\ptensor \unit$.
\end{exa}

In Section~\ref{sec_logspace}, we will consider the algebra spawned by these representations of permutations
as the basic components of logarithmic space programs in $\ualg$.

\begin{defi}[permutation algebra]\label{def_perm}
For $n\in \BB N$ we set $[\F S_n]:=\set{[\sigma]}{\sigma \in \F S_n}$ and $\C S_n$ as the $\ast$-algebra
spawned by $[\F S_n]$.

We define then the \emph{permutation algebra} $\C S$ as the $\ast$-algebra spawned by $\displaystyle\bigcup_{\mathclap{n\in \BB N}} \C S_n$.%
\end{defi}

\section{Words, Observations and Normativity}
\label{sec_words}
The resolution algebra $\ualg$ embeds its own mechanism of execution, unification, and we saw how permutations could be represented in it. %
This is the general environment where the rest of this work is going to take place.

At this stage, there is no distinction between data and programs, functions and inputs.
In this section, we single out two subsets of the algebra: one will represent data, the other corresponds to programs.
In Section~\ref{sec_normativity} we will see how to address through the notion of \emph{normativity} the
fact that many wirings can represent the same data in our algebraic view.
This will lead to the definition of
an acceptance predicate, based on nilpotency:

\begin{defi}[nilpotency]\label{def_nilp}
A wiring $F$ is \emph{nilpotent} if $F^n=0$ for some $n\in \BB N$.
\end{defi}

In the \GofI models
which are the intuitive starting point of this work, this corresponds to strong normalization,
\ie termination of computation.
Note that this makes the acceptance only semi-decidable in
general~\cite{Bagnol2014}:
one can always compute iterations of a wiring and eventually reach $0$, but there is no general algorithm
to decide if a wiring is \emph{never} going to reach $0$. However, we will consider in 
Section~\ref{subsec_soundness} a particularly simple kind of
wirings with an acceptance problem simplified to the point it becomes a logarithmic space problem.

We will consider words on alphabet as our data, although more complex datatypes could be represented
as long as they enjoy a representation in \lcalc, following the same pattern.

 \subsection{Representing Computation: Words and Observations}
\label{sec_words_and_obs}
The representation of words over an alphabet, seen here as a set of constant symbols, in the resolution algebra directly comes from the
translation of the representation of words in \lcalc (or in linear logic) and their
interpretation in \GofI models~\cite{Girard1989a,Girard1995a}.

This proof-theoretic origin is an useful guide for intuition, but we can give a direct
definition of the notion.

\begin{nota}
We fix distinguished constant symbols $\lft$, $\rgt$ and $\start$, with $\start\notin\Sigma$
	and we write $u \sflow v$ the sum $u\flow v+v\flow u$.
	We also let $\unitalg:=\set{\lambda\unit}{\lambda\in \BB C}$ as we did in Definition~\ref{unbounded}.
\end{nota}

\begin{defi}[word representation]\label{words}
From now on we suppose fixed a set $\TT P$ of constant symbols, the \emph{position constants}
and denote with $\TT c_i$ the symbols of the alphabet $\Sigma$.

Let $W=\TT c_1 \dots \TT c_n$ be a word over $\Sigma$ and $\TT p_0,\TT p_1,\dots,\TT p_n$ be distinct
position constants.
The \emph{representation} $W(\TT p_0,\TT p_1,\dots,\TT p_n)$ of $W$
with respect to $\TT p_0,\TT p_1,\dots,\TT p_n$
is an isometric wiring (Definition~\ref{piso}), defined as
\begin{alignat*}{3}
W(\TT p_0,\TT p_1,\dots,\TT p_n) =&& (\start\p\rgt\p x) \p (\TT p_0\p y) &\sflow (\TT c_1\p\lft\p x) \p (\TT p_1\p y) &+ \\
&& (\TT c_1\p\rgt\p x) \p (\TT p_1\p y) &\sflow (\TT c_2\p\lft\p x) \p (\TT p_2\p y) &+ \\
&&& \vdotswithin{\sflow}\\
&& (\TT c_n\p\rgt\p x) \p (\TT p_n\p y) &\sflow (\start\p\lft\p x) \p (\TT p_0\p y)
\end{alignat*}

\end{defi}

In this definition, the position constants $\TT p_0,\TT p_1,\dots,\TT p_n$ can be understood as the
addresses of memory cells holding the symbols $\start, \TT c_1, \dots, \TT c_n$.
This simulates the order naturally present when an input word is written on a tape (where each cell has one or two neighbours)
in a context where the commutative addition cannot implement any form of order:
we henceforth have to \emph{tag} each symbol with a position.

More generally, this
representation of words is to be understood as \emph{dynamic}: we may think of a series of
\emph{movement instructions} from a symbol to the next or the previous for a kind of pointer machine.
This is why each term of the form $\TT c_i \p \dots \p \TT p_i \p \dots$ comes in two distinct versions
using either $\rgt$ or $\lft$ (as \enquote{right/next}, \enquote{left/previous}) setting different responses
to different directions in reading the input word.
Moreover, the general shape of
the wiring is circular, \ie when reaching the end of the word, we return to the position holding $\start$.
This can be pictured as follows:
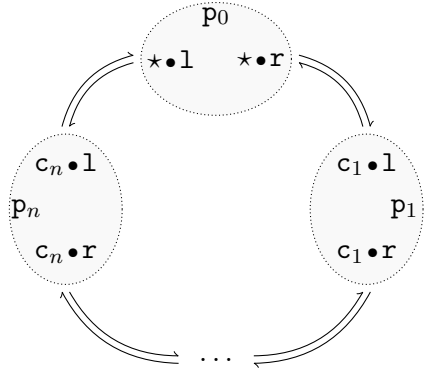
\begin{center}
\label{dessin}
\begin{tikzpicture}
	\draw[densely dotted, fill=gray!5] (0,2) ellipse (1cm and 0.75cm) node[above=0.3, black] {$\TT p_0$};
    \node (star-r) at (0.6, 2) {$\start \p \rgt$};
    \node (star-l) at (-0.6, 2) {$\start \p \lft$};
    \draw[densely dotted, fill=gray!5] (2,0) ellipse (0.75cm and 1cm) node[right=0.15, black] {$\TT p_1$};
    \node (c1-l) at (2, 0.6) {$\TT c_1 \p \lft$};
    \node (c1-r) at (2, -0.6) {$\TT c_1 \p \rgt$};
	\node at (0,-2) {$\hdots$};
	\draw[densely dotted, fill=gray!5] (-2,0) ellipse (0.75cm and 1cm) node[left=0.15, black] {$\TT p_n$};
    \node (cn-l) at (-2, 0.6) {$\TT c_n \p \lft$};
    \node (cn-r) at (-2, -0.6) {$\TT c_n \p \rgt$};
    \draw[-left to] (1.1, 2.05) to [bend left] (2.05, 1.1); %
    \draw[-left to] (1.95, 1.1) to [bend right] (1.1, 1.95); %
    \draw[-left to] (-1.1, 1.95) to [bend right] (-1.95, 1.1); %
    \draw[-left to] (-2.05, 1.1) to [bend left] (-1.1, 2.05); %
    \draw[-left to] (2.05, -1.1) to [bend left] (0.5, -2.05); %
    \draw[-left to] (0.5, -1.95) to [bend right] (1.95, -1.1); %
    \draw[-left to] (-0.5, -2.05) to [bend left] (-2.05, -1.1); %
    \draw[-left to] (-1.95, -1.1) to [bend right] (-0.5, -1.95); %
\end{tikzpicture} \end{center}

This point of view will be at work in the proof of Theorem~\ref{th_completeness}, where we will show that computations of a particular class of pointer machines can be represented in our context.
In that perspective, the $\dots \p x \p \dots y \p \dots$ part will serve to preserve some information relative to the machine, such as its internal state or the positions of additional pointers.

To identify in which $\ast$-algebra all representations of words live, let us define some notations and sub-algebras.

\begin{defi}\label{def_various_alg}
	We write, with a slight abuse of notation that identifies algebras and sets of symbols generating them,
	\begin{itemize}
	\item $\Sigma_{\lft\rgt}$ the $\ast$-algebra generated by flows of the form
	$\TT s\p\TT d\flow \TT s'\p\TT d'$ with $\TT s,\TT s'\in \Sigma\cup\{\start\}$ and $\TT d,\TT d'\in \{\lft,\rgt\}$
	\item $\TT P$ the $\ast$-algebra generated by flows of the form
	$\TT p\flow \TT p'$ with $\TT p,\TT p'\in \TT P$,
	\item $\TT Q$ the $\ast$-algebra generated by flows of the form,
	$\TT q\flow \TT q'$ with $\TT q,\TT q'$ %
being \emph{state constants}, whose set is denoted $\TT Q$.
	\end{itemize}
	
\end{defi}

\begin{prop}
	Any representation of a word $W(\TT p_0,\TT p_1,\dots,\TT p_n)$ lies in the $\ast$-algebra
	\[\mathcal W=\Sigma_{\lft\rgt}\ptensor\unitalg\ptensor\TT P\tpower1\]
	which we call the \emph{word algebra}. Moreover word representations are concrete wirings.
\end{prop}

Let us turn now to the definition of \emph{observations} that will correspond to programs computing
on representations of words.

We give a general notion first, which we will instantiate in Section~\ref{sec_logspace} to get a class
of observations that captures
logarithmic space
computation, based on the representation of permutations over an unbounded tensor presented in
the previous section.

\begin{defi}[observations]\label{def_obs}
Given a $\ast$-algebra $\C A$, an \emph{observation by $\C A$} is an element of $\C O[\C A]^{\concrete}$ where
\[\C O[\C A]=\Sigma_{\lft\rgt}\ptensor\C A\] %
and the $(\cdot)^{\concrete}$ notation refers to the set of concrete (Definition~\ref{concrete})
wirings obtained from $\C O[\C A]$.
Moreover when an observation by $\C A$ happens to be an isometric wiring, we will call it an \emph{isometric observation} by $\C A$.

In case $\C A$ is the whole unification algebra $\ualg$ we call the elements of $\C O[\ualg]^{\concrete}$
simply \emph{observations}.
\end{defi}
 \subsection{Independence from Representations: Normativity}
\label{sec_normativity}
We now study how representations of words and observations interact, leading to a notion of acceptance.
The basic idea is that an observation $\phi$ accepts a word $W$ if the wiring
$\phi W(\TT p_0,\dots,\TT p_n)$ is nilpotent (Definition~\ref{def_nilp}), but we want to make sure that the notion is independent of
the choice of a specific representation of $W$.

This could be enforced at a basic syntactic level: we could forbid the observations to use the position constants,
which are the only source of variability from one representation to the other. %
But we would like to give a more algebraic view of this idea,
to tend to a more abstract vision: this leads to the notion of \emph{normative pair}
introduced by J.-Y.~Girard~\cite{Girard2012}.

\begin{defi}[automorphism]
An \emph{automorphism} of a $\ast$-algebra $\C A$ is an injective linear application $\varphi: \C A\rightarrow\C A$
such that for all~$F,G\in\C A$:
\[\varphi(FG)=\varphi(F)\varphi(G) \quad \text{ and } \quad \varphi(F^\dagger)=\varphi(F)^\dagger\]
\end{defi}

For instance $\varphi(U_1\ptensor U_2)=U_2\ptensor U_1$ defines an automorphism of $\ualg \ptensor\ualg$.

An automorphism $\varphi$ therefore preserves the algebraic properties of elements of $\C A$: in particular,
$\varphi(A)$ is nilpotent if and only if $A$ is.

\begin{nota}
If $\varphi$ is an automorphism of $\C A$ and $\psi$ is an automorphism of $\C B$, we write
$\varphi\ptensor\psi$ the automorphism of $\C A\ptensor\C B$ defined for all $A\in \C A, B\in \C B$ as
\[(\varphi\ptensor\psi)(A\ptensor B)=\varphi(A)\ptensor\psi(B)\]
\end{nota}

Keeping in mind that an automorphism is a transformation that preserves the algebraic
properties, we define the notion of normative pair as a pair of $\ast$-algebras such that an automorphism
of one of them can be extended to act as the identity on the other one.

\begin{defi}[normative pair]
\label{normpair}
A pair $(\C A,\C B)$ of $\ast$-algebras is a \emph{normative pair} whenever any automorphism $\varphi$
of $\C A$ can be extended into an automorphism $\overline\varphi$ of the $\ast$-algebra $\C E$
generated by $\C A \cup \C B$ such that $\overline\varphi(B)=B$ for any $B\in \C B\subseteq\C E$.
\end{defi}

A trivial example would be that of commuting $\C A$ and $\C B$: then any element of $\C E$ can be
written as a sum of $AB$ with $A\in\C A$ and $B\in\C B$, which allows us then to define
$\overline\varphi(AB)=\varphi(A)B$ consistently.
But this case is of little interest since when
$A$, $B$ commute $(AB)^n=A^nB^n$ so that there is no real interaction between $A$ and $B$: they
\enquote{pass through} each other without communicating.

The two following propositions set the basis for a notion of acceptance and rejection independent of the representation of a word, as soon as normative pairs are involved.

\begin{prop}[automorphic representations]
\label{automorphic}
Any two representations $W(\TT p_0,\dots,\TT p_n)$, $W(\TT p'_0,\dots,\TT p'_n)$ of the same word $W$
are automorphic: there is an automorphism $\varphi$ of $\unitalg \ptensor \TT P\tpower 1$ such that
\[(\Id_{\Sigma_{\lft\rgt}}\ptensor\varphi)\big(W(\TT p_0,\dots, \TT p_n)\big)=W(\TT p'_0,\dots, \TT p'_n)\]
\end{prop}

\proof
Consider any bijection $f:\TT P \rightarrow\TT P$ such that $f(\TT p_i)=\TT p'_i$ for all $i$.

Then set $\varphi(x\p v\p y\flow x\p w\p y)=x\p f(v)\p y\flow x\p f(w)\p y$, extended by linearity.
\qed

\begin{prop}[nilpotency and normative pairs]
\label{normative}
Let $(\C A,\C B)$ be a normative pair and $\varphi$ an automorphism of $\C A$.
Let $F\in \Sigma_{\lft\rgt}\ptensor\C A$, $G\in \Sigma_{\lft\rgt}\ptensor\C B$ and $\psi=\Id_{\Sigma_{\lft\rgt}}\ptensor \varphi$.

Then $GF$ is nilpotent if and only if $G\psi(F)$ is nilpotent.
\end{prop}

\proof
Let $\overline\varphi$ be the extension of $\varphi$ as in Definition~\ref{normpair}
and $\overline\psi=\Id_{\Sigma_{\lft\rgt}}\ptensor\overline\varphi$.

We have for all $n\neq 0$ that $(G\psi(F))^n=(\overline\psi(G)\overline\psi(F))^n=(\overline\psi(GF))^n =\overline\psi((GF)^n)$.

By injectivity of $\overline\psi$, $(G\psi(F))^n=0$ if and only if $(GF)^n=0$.
\qed

In view of this last proposition, as we know that words are in
$\mathcal W=\Sigma_{\lft\rgt}\ptensor(\unitalg\ptensor\TT P\tpower1)$ and observation by $\C B$ are in
$\C O[\C B]^{\concrete} = (\Sigma_{\lft\rgt}\ptensor \C B)^{\concrete}$,
we understand that it is enough that
$(\unitalg\ptensor\TT P\tpower1,\C B)$ constitutes a normative pair to get the expected result.

\begin{cor}[independence]
\label{indep}
If $(\unitalg\ptensor\TT P\tpower1,\C B)$ is a normative pair, $W$ a word and
$\phi$ an observation by $B$, the product $\phi W(\TT p_0,\dots, \TT p_n)$ is nilpotent for one choice of
$(\TT p_0,\dots,\TT p_n)$ if and only if it is nilpotent for all choices of $(\TT p_0,\dots,\TT p_n)$.
\end{cor}

In the next section, we will consider a particular case of observations, namely where
$\C B=\TT Q \ptensor \C S$, where $\TT Q$ is the algebra generated by the flows of state constants from Definition~\ref{def_various_alg} and $\C S$ is the permutation
algebra from Definition~\ref{def_perm}.

\begin{thm}[normativity]\label{th_norm}
For any $\C A,\C C$ the pair $\big(\unitalg\ptensor\C A\tpower1,\C C \ptensor\C S\big)$ is normative.
In particular, $(\unitalg\ptensor\TT P\tpower1,\TT Q \ptensor \C S)$ is normative.
\end{thm}

\proof

Consider $\varphi$ is an automorphism of $\unitalg\ptensor \C A\tpower 1$.
It can be written as
\[\varphi(\unit\ptensor G\ptensor \unit)=\unit\ptensor\psi(G)\ptensor \unit\]
for all $G$, with $\psi$ an automorphism of $\C A$.
Now, the $\ast$-algebra generated by $\unitalg\ptensor\C A\tpower1$ and $\C C \ptensor\C S$ can be identified as
finite sums of elements of the set
\[\set{c \ptensor \sigma F}{c \in \C C, \sigma \in \C S \text{ and } F\in \C A\tpower\infty}\]
This is because for any $\sigma \in \C S$ and $F \in A^{\otimes \infty}$ we have $F\sigma=\sigma F'$ for a suitable $F'\in A^{\otimes \infty}$ obtained by letting $\sigma$ swap components of $F$.

We set for $F=F_1\ptensor\cdots\ptensor F_n\ptensor \unit\in\C A\tpower n$,
\[\tilde\varphi(F)=\psi(F_1)\ptensor\cdots\ptensor \psi(F_n)\ptensor \unit\]
which extends into an automorphism of $\C A\tpower \infty$ by linearity.
Finally, we extend $\tilde\varphi$ to the algebra generated by $\unitalg\ptensor\C A\tpower1$ and $\C C \ptensor\C S$
as $\overline\varphi(c \ptensor\sigma F)=c \ptensor\sigma\tilde\varphi(F)$.
\qed

\begin{rem}
\label{rem-normativity}
This result is likely to be generalized: the permutation algebra \emph{acts} on the infinite tensor product, and through this action
the whole tensor product $\C A\tpower\infty$ is generated by $\C A\tpower1$.
With some adjustment it should be possible to show that given such a situation, one always get a normative
pair.
\end{rem}

We can then define the notion of the language recognized by an observation, thanks to Corollary~\ref{indep} that makes it insensitive to a particular choice of position constants.

\begin{defi}[language of an observation]
\label{def-language-obs}
Let $\phi$ be an observation by $\C B$ satisfying the hypothesis of Corollary~\ref{indep}, \ie $\C B$ is of
the form $\C C \ptensor\C S$.
The \emph{language recognized by $\phi$} is the following set:%
\[\lang(\phi)=\set{W \text{ word over }\Sigma}{\phi W(\TT p_0,\dots,\TT p_n) \text{ nilpotent for any }(\TT p_0,\dots,\TT p_n)}\]
\end{defi}
 
\section{Wirings and Logarithmic Space}
\label{sec_logspace}
Now that we have defined our framework and showed how observations compute, we fix a specific class
of observations %
and study the complexity
of deciding whenever such %
one of its member accepts a word (Section~\ref{subsec_soundness}).
We then show in Section~\ref{subsec_completness} that
the languages recognized correspond exactly, depending on the isometricity of the observation (Definition~\ref{piso}), to the \Logspace or \NLogspace (written \NorDLogspace if we do not want to be specific).

\begin{defi}[\lobs]
\label{def_lobs}
	We consider the $\ast$-algebra $\TT Q \ptensor \C S$ as in Theorem~\ref{th_norm} and call
	\lobs the observations by $\TT Q \ptensor \C S$ (Definition~\ref{def_obs}). More explicitely,
	\lobs{}s are finite sums of flows of the form
	\[(\TT c' \p \TT d' \p \TT q'\flow \TT c\p\TT d\p\TT q)\ptensor[\sigma]\]
	where $\TT c,\TT c'\in\Sigma$, $\TT d,\TT d'\in\IO$, $\TT q,\TT q'\in\TT Q$ and $\sigma$ is a permutation.
\end{defi}

We already shown that the notion of acceptance (Definition~\ref{def-language-obs}) is well-defined since
$(\unitalg\ptensor\TT P\tpower1,\TT Q \ptensor \C S)$ is a normative pair.

Deciding nilpotency in this specific case amounts to build a finite-dimensional vector space where we
can observe the \enquote{relevant computation} taking place in a finitary way (remember that the nilpotency
problem is only semi-decidable in general).
We introduce the notion of \emph{separating space} (Definition~\ref{def-sep-space}) and give a logarithmic space-algorithm based on this notion.

That any \NorDLogspace language, or predicate, can be decided by a \lobs will be proven
thanks to \emph{pointer machines} (Definition~\ref{def-pointer-machines}), a model of computation
designed to be easily encoded.
This model comes as a rephrasing of read-only Turing machines, or more precisely as a modification
of two-way multi-head finite automata, known to capture \NorDLogspace~\cite{Hartmanis1972,Aubert2015Dice}.
Unification will act as a \enquote{hard-wired} way to represent execution, thanks to a dialogue between the representation of the input and the observation.

\subsection{Soundness of Observations}
\label{subsec_soundness}

The aim of this subsection is to prove the following theorem:

\begin{thm}[space soundness]
\label{soundness}
Let $\phi$ be \lobs, its language $\lang(\phi)$ is decidable in \NLogspace.
If moreover $\phi$ is isometric, then $\lang(\phi)$ is decidable in \Logspace.
\end{thm}

The proof is given later on, p.~\pageref{proof_soundness}, but one should notice that the result stands for the complements of these languages, but as $\NLogspace = \coNLogspace$ by the Immerman-Szelepcsényi~\cite{Immerman1988a,Szelepcsenyi1988} theorem, this makes no difference.
Indeed, if one looks closely to the definition of the language of an observation (Definition~\ref{def-language-obs}), one may notice that acceptation rests on the nilpotency of the wiring, which supposes that \emph{all branches of computation ends}: observation looks like a \enquote{universally non-deterministic} model of computation.

This theorem will require the notion of \emph{computation space}: finite dimensional\footnote{We recall that a vector space is finite dimensional if the cardinal of its basis is finite.} subspaces of 
the vector space spawned by closed terms $\closedh$ (Definition~\ref{act_vect}) on which we will be able to observe all the behaviour of certain wirings.
It can be understood as the place where all the relevant computation takes place, which
we call a \emph{separating space}.

The rest of this subsection is devoted to the introduction of this notion of computation space, and to prove that a computation space can be computed from an observation and a word, and is indeed separating.
Finally, we prove that deciding if a computation space is nilpotent\incise{which is equivalent to an observation applied to an input being nilpotent}can be done with logarithmic space, thus proving Theorem~\ref{soundness}.

\begin{defi}[separating space]
\label{def-sep-space}
A subspace $E$ of $\closedh$ is \emph{separating} for a wiring $F \in \closedh$ if $F(E)\subseteq E$ and if $F^k(E)=0$ implies $F^k=0$.
\end{defi}

If we observe the computation of $F$ on $E$, it cannot \enquote{step outside $E$}.
On the other hand, the fact that $F^k(E)=0$ implies $F^k=0$ means that it is enough to
check that a certain iteration of $F$ cancels on the space $E$ to conclude that it cancels
everywhere.

\begin{defi}[computation space]
\label{compspace}
Let $\{\TT p_0,\dots, \TT p_n\}$ be a set of distinct position constants and $\phi$ a \lobs.
Let $N(\phi)$ be the smallest integer and $\TT S(\phi)$ the smallest (finite) subalgebra of $\TT Q$
such that $\phi \in \Sigma_{\lft\rgt}\ptensor\TT S(\phi)\ptensor\C S_{N(\phi)}$:

The \emph{computation space} of $\phi$ associated to
the positions $\TT p_i$, denoted $\comp_\phi (\TT p_0,\dots, \TT p_n)$, is the subspace of $\closedh$
generated by closed terms of the form
\[\TT c\p \TT d \p \TT q \p(a_1\p\cdots\p a_{N(\phi)}\p \start)\]
where $\TT c\in \Sigma \cup \{\start\}$, $\TT d\in\IO$, $\TT q\in \TT S(\phi)$ and $\forall 1 \leqslant i \leqslant N(\phi)$, $a_i\in \{\TT p_0,\dots, \TT p_n\}$.

Denoting $|A|$ the cardinal of $A$, the dimension of $\comp_\phi (\TT p_0,\dots, \TT p_n)$ is
\[(|\Sigma| + 1) \times 2 \times |\TT S(\phi)| \times (n+1)^{N(\phi)}\]
which is polynomial in $n$.
\end{defi}

\begin{lem}[separation]
\label{sep}
For any \lobs $\phi$ and any word $W$, %
$\comp_\phi(\TT p_0,\dots, \TT p_n)$ is separating for the wiring $\phi W(\TT p_0,\dots, \TT p_n)$.
\end{lem}

\proof
Immediate given how $\comp$ was defined.
\qed

Now we proceed to prove Theorem~\ref{soundness}: %
we provide an algorithm deciding the nilpotency of $\phi W(\TT p_0,\dots, \TT p_n)$ in logarithmic
space, based on the above lemma.

\proof \emph{(of Theorem~\ref{soundness})}
\label{proof_soundness}
We define the non-deterministic algorithm below.
It takes as an input a word $W$ of length $n$.
We remark that the observation $\phi$ being a constant, one can compute once and for all $N(\phi)$ and $\TT S(\phi)$.

\begin{multicols}{2}
\begin{algorithmic}[1]
\STATE $D\gets (|\Sigma| + 1) \times 2 \times |\TT S(\phi)| \times (n+1)^{N(\phi)}$
\STATE $C\gets 0$
\STATE pick a term $v\in\comp_\phi(\TT p_0,\dots, \TT p_n)$\label{pick1}
\WHILE{$C\leq D$}
\IF{$(\phi W(\TT p_0,\dots, \TT p_n))(v)=0$} \label{if}
\RETURN ACCEPT
\ENDIF
\STATE pick a term $v' \in (\phi W(\TT p_0,\dots, \TT p_n))(v)$ \label{pick2}
\STATE $v\gets v'$ \label{algo-erase}
\STATE $C\gets C+1$
\ENDWHILE
\RETURN REJECT
\end{algorithmic}
\end{multicols}

All computation paths (the \enquote{pick} at lines~\ref{pick1} and~\ref{pick2} being non-deterministic choices) accept if and only if $(\phi W(\TT p_0,\dots, \TT p_n))^n(\comp_\phi(\TT p_0,\dots, \TT p_n))=0$ for some $n$ lesser or equal to the dimension $D$ of the computation space $\comp_\phi(\TT p_0,\dots, \TT p_n)$.
By Lemma~\ref{sep}, this is equivalent to $\phi W(\TT p_0,\dots, \TT p_n)$ being nilpotent, as the computation space is a separating space for $\phi W(\TT p_0,\dots, \TT p_n)$.

The terms chosen at lines~\ref{pick1} and~\ref{pick2} are representable by an integer of size at most $D$, and we need to store only two such terms at the same time, as one is replaced by the other at line~\ref{algo-erase}, every time we go through the \textbf{while}-loop.
We already mentioned that the dimension $D$ of the computation space is polynomial in the size of the input (Definition~\ref{compspace}).
As $C$ is bounded by $D$, both integers can be stored in a space logarithmic in $n$.

Computing $(\phi W(\TT p_0,\dots, \TT p_n))(v)$ at line~\ref{if} can be performed by matching
$v$ against the list of flows in $W(\TT p_0,\dots, \TT p_n)$ then $\phi$. This can be done within
logarithmic space since all the terms involved are of size at most logarithmic, and matching
can be computed in linear space (remember the discussion below Theorem~\ref{thm-unif-logspace}).

Moreover, if $\phi$ is an isometric wiring, $(\phi W(\TT p_0,\dots, \TT p_n))(v)$ consists of a single term instead of a sum by Lemma~\ref{lem_isom}, and there is therefore no non-deterministic choice to be made at line~\ref{pick2}.
It is then enough to run the algorithm enumerating all terms of $\comp_\phi(\TT p_0,\dots, \TT p_n)$ at line~\ref{pick1} to determine the nilpotency of $\phi W(\TT p_0,\dots, \TT p_n)$.
\qed

We can have a more graph-oriented view of the algorithm.
Picture the elements of $\comp_\phi(\TT p_0,\dots, \TT p_n)$ as vertices of a graph:
they represent all the possible terms the computation could reach.
The wiring induces the edges between those terms: if $u$ is in the image of $v$
then they both belong to $\comp_\phi(\TT p_0,\dots, \TT p_n)$ and
we can draw an edge between them.
The iteration of this procedure gives a set of reachable terms, the trace of all possible computation.
To know if this wiring is nilpotent, one only has to check whether there is a cycle in the graph obtained.
This is a typical \Logspace problem, and that this graph can be built in \Logspace manly rests on
the fact that matching is in \Logspace too. The algorithm above performs both tasks (building the graph
and looking for cycles) at the same time.
\subsection{Completeness: Representing Pointer Machines as Wirings}
\label{subsec_completness}
To prove the converse of Theorem~\ref{soundness}, we will prove that wirings can encode a special kind of read-only multi-head Turing Machine: pointers machines.
The definition of this model will be guided by our understanding of the wirings' way of computing: they do not have the ability to write or to store information, and acceptance will be defined as termination of all paths of computation.

For a survey of this topic, one may consult the first author's thesis~\cite[Chapter~4]{Aubert2013b},
the main novelty of this part of our work is to notice that reversible computation is represented by isometric operators.

\begin{defi}[pointer machine]
\label{def-pointer-machines}
A \emph{pointer machine} over an alphabet $\Sigma$ is a tuple $(N,\TT S,\Delta)$ where
\begin{itemize}
\item $N\neq 0$ is an integer, the \emph{number of pointers},
\item $\TT S$ is a finite set, the \emph{states} of the machine,
\item $\Delta \subseteq (\Sigma\times\IO\times\TT S)\times(\Sigma\times\IO\times\TT S)\times \F S_N$, are the \emph{transitions} of the machine

(we will write the transitions $(\TT c,\TT d,s) \rightarrow (\TT c',\TT d',s') \times \sigma$, for readability).
\end{itemize}
A pointer machine will be called \emph{deterministic} if for any $A \in \Sigma\times \IO\times\TT S$,
there is at most one $B\in \Sigma\times \IO\times\TT S$ and one $\sigma\in \F S_N$ such that $A\rightarrow B \times \sigma\in\Delta$.
In that case we can see $\Delta$ as a partial function, and we say that the pointer machine is \emph{reversible} if $\Delta$ is a partial injection.
\end{defi}

We call the first of the $N$ pointers the \emph{main} pointer, it is the only one that can move (it will be moved by the representation of the integer, as we shall see below).
The other pointers are referred to as the \emph{auxiliary} pointers.
An auxiliary pointer will be able to become the main pointer during the computation thanks to permutations.

\begin{defi}[configuration]
Given the length $n$ of a word $W=\TT c_1\dots \TT c_n$ over $\Sigma$ and a pointer machine $M=(N,\TT S,\Delta)$, a \emph{configuration} $C$ of $(M,n)$ is an element of 
\[\Sigma\times\IO\times \TT S\times\{0,1,\dots,n\}^N.\]
\end{defi}

The element of $\TT S$ is the state of the machine and the element of $\Sigma$ is the symbol 
the main pointer points at.
The element of $\IO$ is the direction of the next move of the main pointer, and the elements of $\{0,1,\dots,n\}^N$ correspond to the positions of the (main and auxiliary) pointers on the input.

As the input tape is considered cyclic with a special symbol marking the beginning of the word (recall Definition~\ref{words}), the pointer positions are \emph{modulo} $n+1$ integers for an input word of length $n$.

\begin{defi}[transition]
Let $W$ be a word of length $n$ and $M=(N,\TT S,\Delta)$ be a pointer machine.
A \emph{transition} of $M$ on input $W$ is a triple of configurations
\[\TT c,\TT d,s,(p_1,\dots,p_N) \trans{\move}{} \TT c',\overline{\TT d},s,(p_1',\dots,p_N') \trans{\swap}{} \TT c'',\TT d',s',(p_{\sigma(1)}',\dots,p_{\sigma(N)}') \]
such that
\begin{enumerate}
\item if $\TT d\in\IO$, $\overline{\TT d}$ is the other element of $\IO$,
\item $p_1'=p_1-1$ if $\TT d=\lft$ and $p_1'=p_1+1$ if $\TT d=\rgt$,
\item $p_i'=p_i$ for $i\neq 1$,
\item $\TT c$ (resp. $\TT c'$) is the symbol at position $p_1$ (resp. $p_1'$) of $W$, \label{condition}
\item and $(\TT c',\overline{\TT d},s) \rightarrow (\TT c'',\TT d',s') \times \sigma$ belongs to $\Delta$.
\end{enumerate}
\end{defi}
There is no constraint on $c''$, but every time this value differs from the symbol pointed by $p_{\sigma(1)}'$, the computation will halt on the next \move phase, because there is a mismatch between the value that is supposed to have been read and the actual symbol
of $W$ stored at this position, and that would contradict the first part of item~\ref{condition}.

In terms of wirings, the \move phase corresponds to the application of the representation of the word, whereas the \swap phase corresponds to the application of the observation.
One way to present it is to draw attention on the fact that, in the drawing page~\pageref{dessin}, there is no arrow \emph{inside} the \enquote{domain} of a position constant.
Stated differently, \emph{observations only} can make the computation evolves from a direction constant to another, and \emph{representation of the word only} can update the constant position.

\begin{defi}[acceptance]
\label{translate}
A pointer machine $M$ accepts a word $W$ of length $n$ if %
for all configuration $C_0$ of $(M,n)$,
all sequences of transitions
\[\big(C_0\trans{\move}{} C_0' \trans{\swap}{} C''_0 = C_{1} \trans{\move}{} \cdots \trans{\swap}{} C''_{k-1} = C_{k} \big)\]
are finite.
We write $\lang(M)$ the set of words accepted by $M$.
\end{defi}

This means informally that a pointer machine accepts a word if it cannot ever loop, from whatever configuration it starts from.
That a lot of paths of computation stop and accept \enquote{wrongly} is no worry, since only rejection is meaningful: our pointer machines compute in a \enquote{universally non-deterministic} way, to stick to the acceptance condition of wirings, nilpotency.

\begin{prop}[space and pointer machines]
\label{pointerl}
If $L\in \NLogspace$, then there exist a pointer machine $M$ such that $\lang(M)=L$.
Moreover, if $L \in \Logspace$ then $M$ can be chosen to be reversible.
\end{prop}

\proof
We proceed with a step-by-step transformation from Turing machines deciding $L$ to pointer machines deciding the same language, through automata.

If $L \in \NLogspace$, then by the Immerman-Szelepcsényi~\cite{Immerman1988a,Szelepcsenyi1988} theorem there exists a Turing machine in $\coNLogspace$ that decides it.
If $L \in \Logspace$, as deterministic and reversible logarithmic space coincides~\cite{Lange2000}, we take a reversible Turing machine that decides it.

It is then possible to design two (non-)deterministic multi-head finite automata that recognize the same languages~\cite{Hartmanis1972}.
Those automata are read-only version of the Turing machines that are closer to pointer machines, but still differ on some features.
We now prove that they can be re-arranged to fit the definition of pointer machines.

First, we modify their accepting condition to be \enquote{halt} and their rejection to be \enquote{loop}: it is always possible to adjust automata so that no transition is defined from an accepting state, and so that rejection makes the automata never halts.
The introduction of loops requires some care: an \enquote{in-place loop} would prevent the pointer machine to ever stop no matter the input, so we implement loops thanks to a \enquote{re-initialization} (\enquote{go back to an initial state, with all your heads at position $\TT p_0$, reading $\star$}).
Looping is in that setting more of a \enquote{check forever that you do not halt on that input}~\cite[Section~5.1]{Aubert2016IaC}\cite[Section~6.2.3]{Aubert2016mscs}.
In the non-deterministic case, this amounts to define acceptance as \enquote{all branches of computation halt}, and rejection as \enquote{at least one branch never halts}.

Then, we transform the transition function so that at most one head moves at each transition.
But this is not enough: the transition function of our pointer machines reads only one symbol at a time, the one pointed by the main pointer.
We operate a sort of currying to the transition function of the automata: it reads only one symbol at a time, and the other symbols that were read by the other heads as well as the direction of their last move are encoded in the states.

Finally, we re-arrange the automata so that swapping the heads and moving them on the input are two different phases.
We obtained a pointer machine.
\qed

The first author recently proposed a recollection of the classical characterizations of complexity classes by automaton~\cite{Aubert2015Dice}.
Once the basics tricks of the transformation of a $\NorDLogspace$-Turing machine into a two-ways multi-head automata are known, 
it becomes easy to use some classical theorems to get a pointer machine.
The pointer machine obtained has more pointers than the automata had heads, and the number of state grew violently, but independently from the input, and without losing computational power.

As our pointer machines are designed to be easily simulated by wirings, we get the expected result almost for free.

\begin{thm}[space completeness]
\label{th_completeness}
If $L\in \NLogspace$, then there exist a \lobs $\phi$ such that $\lang(\phi)=L$.
Moreover, if $L \in \Logspace$ then $\phi$ can be chosen isometric.
\end{thm}

\proof
There exists a pointer machine $M=(N,\TT S,\Delta)$ such that $\lang(M) = L$ by Proposition~\ref{pointerl}.
We associate to the set of states $\TT S$ a set of constants that we still write $\TT S$.
To any element $D=(\TT c,\TT d,s) \rightarrow (\TT c',\TT d',s') \times \sigma$ of $\Delta$ we associate the flow 
\[[D]=(\TT c'\p\TT d'\p s' \flow \TT c\p\TT d\p s) \ptensor [\sigma]\]
which belongs to $\Sigma_{\lft\rgt}\ptensor\TT Q\ptensor\C S_n$ and we define the \lobs
$[M]$ as $\displaystyle\sum_{D\in \Delta} [D]$.

One can easily check that this translation preserves the language recognized (there is even a step by step simulation of the computation on the word $W$ by the wiring $[M]W(\TT p_0,\dots, \TT p_n)$) and relates reversibility with isometricity: in fact, $M$ is reversible if and only if $[M]$ is an isometric wiring.
Then, if $L\in \Logspace$, $M$ is deterministic and can always be chosen to be reversible~\cite{Lange2000}.
\qed

\section*{Conclusion}
\subsection*{Related Works}
\label{related-works}
The idea to consider the \GofI representation of integers with implicit complexity perspectives is originally due to Girard~\cite{Girard2012}, where one of his motivation was to prove that no representation of the integers was \enquote{more standard} than any other (the \enquote{normativity} theorem). %
This approach diverges from the usual complexity results coming from linear logic, that entail a bound on complexity by restricting programs thanks to a type system.
In this perspective, limitation on the computational power of observations (representations of programs) comes from \emph{algebraic restrictions}.

A first series of work~\cite{%
Seiller2012,%
Aubert2013b,%
Aubert2016mscs,%
Aubert2016IaC%
} deepened those intuitions by making formal the interpretation of proofs of linear logic in the hyperfinite factor, a type II$_1$ von Neuman Algebra.
Thanks to the representation of infinite operators by matrices, it was proven~\cite{%
Seiller2012,%
Aubert2016mscs%
} that representations of programs in a specific sub-algebra were characterizing \NLogspace.
Later on, an additional restriction on the observations, phrased in terms of norm, was proven~\cite{Aubert2013b,%
Aubert2016IaC%
} to characterize \Logspace.
As an additional result, it was also discovered that the observations' mechanism of computation was deeply related to automata theory~\cite{Aubert2013b,%
Aubert2016IaC%
}.

The present work and its previous version~\cite{Aubert2014%
} constitute a bridge between this algebraic setting and a more syntactical one that followed.
It still uses ad-hoc \enquote{pointer machines} and a full algebra to describe the computation of observations.
In that perspective, normativity is still seen as mathematical feature, i.e.\ the existence of automorphisms to switch from a representation to another without affecting observations, whereas it is a plain $\alpha$-conversion in the following works.

Two changes in the perspective appeared later on: first, this whole construction could be rephrased in the latest formulation of \GofI~\cite{Girard2013}, which fully uses first-order terms and unification, or matching, to represent linear logic and its execution procedure.
This more syntactical presentation allows to isolate restriction on terms as syntactical conditions: \emph{balanced}~\cite{Aubert2014b%
} and \emph{unary}~\cite{Aubert2016fossacs%
} flows were proven to characterize respectively \NorDLogspace and \Ptime.
Those innovative limits imposed on wiring were discovered thanks to a careful attention paid to automata theory, which is the second change in the perspective.
The rephrasing of the \emph{memoization technique}~\cite{Cook1971}\incise{that was originally invented to prove the \Ptime-soundness of pushdown automata}applied to flows permitted to get the first time-bounded characterization of a complexity class in a \GofI setting.
Those works benefited from previous characterizations in terms of \emph{unification algebra}~\cite{Baillot2001}, and constitutes a modern rephrasing of this work.
The algebraic framework is lighter, for it uses semi-ring~\cite{Bagnol2014} rather than full algebras.

As a by-product, this series of works is now closer to logic programming~\cite{Bagnol2014,%
Aubert2014b,%
Aubert2016fossacs%
} and pretends to highlight with new perspectives this subject.

\subsection*{Future Directions}
We built an algebra endowed with an evaluation mechanism relying on unification of first-order terms, that allows to seamlessly represent the execution of programs.
Taking as a guiding intuition the functional representation of data as functions, we took the Curry-Howard interpretation of \Lterms as proofs to divide our algebra between inputs and observations.
We separated them in two different sub-algebras that communicate \enquote{just enough}: the input cannot interfere with the observation, the observation is insensitive to the choice of the representation (this is the normativity property), and yet they can represent decision of predicates.
Using a specific sub-space that represents all possible computations, the computation space, we proved that deciding the outcome of the interaction between an observation and its input was decidable in \NorDLogspace.
This models, thanks to the representation of permutations and unbounded tensor product, has enough computational power to represent \NorDLogspace.
To prove it, we had to pay an extra attention to the peculiarities of this model: it computes as a read-only model, whose heads read and move in a restricted way, and who accepts by halting.
It was nevertheless possible to introduce pointer machines, mid-way between observations and two-ways multi-head automata, and to prove that observations could simulate this \NorDLogspace-model of computation.

The language of the unification algebra gives a twofold point of view on computation, either through algebraic structures or pointer machines.
We may therefore start exploring possible variations of the construction, combining intuitions from both worlds.

The algebraic setting allows for a number of modifications whose computational meaning is still unclear.
We considered only the computational features provided by concrete wirings, but 
one could imagine that negative coefficients would provide a mechanism to interact \enquote{at distance} between branches of computation.
That may offer the opportunity to represent parallel computation with a mechanism of synchronization, a branch of computation being able to get cancelled by another one.
It is also worth mentioning that matrix computation is well-known to relate closely to parallel computation: observations could be evaluated in a parallel setting, providing complexity-bound on time rather than on space.

Pointer machines relate closely to automata theory, which is a vivid research field that should be inspiring.
Apart from the \Ptime-characterization provided by pushdown automata that was already explored~\cite{Bagnol2014,%
Aubert2016fossacs%
}, relations between our setting and pushdown systems, tree automata or asynchronous automata definitely ought to be studied.
This angle could also provide intuitions to tackle the switching from complexity of
predicates to complexity of functions, using transducers instead of automata.

\subsection*{Acknowledgement}
The authors would like to thanks Jean-Yves Girard for inspiring hints he gave us during the writing of this article.
They are also deeply grateful to Paolo Pistone and Thomas Seiller, with whom we pushed forward this promising line of work.
Discussions with them modified our way of presenting this work, and we are much obliged.
Finally, we would like to thank the reviewers of the conference and journal versions of this work for their insightful comments.

\bibliographystyle{alpha}	%

\end{document}